\documentclass[english]{IEEEtran}
\usepackage[T1]{fontenc}
\usepackage[latin9]{inputenc}
\usepackage{bm}
\usepackage{amsthm}
\usepackage{amsmath}
\usepackage{amssymb}

\makeatletter
\theoremstyle{plain}
\newtheorem{thm}{\protect\theoremname}
\theoremstyle{definition}
\newtheorem{defn}[thm]{\protect\definitionname}
\theoremstyle{plain}
\newtheorem{prop}[thm]{\protect\propositionname}
\theoremstyle{remark}
\newtheorem{rem}[thm]{\protect\remarkname}
\theoremstyle{plain}
\newtheorem{lem}[thm]{\protect\lemmaname}
\theoremstyle{plain}
\newtheorem{cor}[thm]{\protect\corollaryname}

\makeatother

\usepackage{babel}
\providecommand{\corollaryname}{Corollary}
\providecommand{\definitionname}{Definition}
\providecommand{\lemmaname}{Lemma}
\providecommand{\propositionname}{Proposition}
\providecommand{\remarkname}{Remark}
\providecommand{\theoremname}{Theorem}

\begin{document}

\title{Technical Report: Observability of a Linear System under Sparsity
Constraints}

\author{Wei Dai and Serdar Yüksel}
\maketitle
\begin{abstract}
Consider an $n-$dimensional linear system where it is known that
there are at most $k<n$ non-zero components in the initial state.
The observability problem, that is the recovery of the initial state,
for such a system is considered. We obtain sufficient conditions on
the number of the available observations to be able to recover the
initial state exactly for such a system. Both deterministic and stochastic
setups are considered for system dynamics. In the former setting,
the system matrices are known deterministically, whereas in the latter
setting, all of the matrices are picked from a randomized class of
matrices. The main message is that, one does not need to obtain full
$n$ observations to be able to uniquely identify the initial state
of the linear system, even when the observations are picked randomly,
when the initial condition is known to be sparse.\renewcommand{\thefootnote}{\fnsymbol{footnote}} 
\setcounter{footnote}{-1}%
\footnote{Dr. Wei Dai is with Department of Electrical and Electronic Engineering,
Imperial College London, London SW7 2AZ, United Kingdom. Email: wei.dai1@imperial.ac.uk

Prof. Serdar Yüksel is with Mathematics and Engineering Program, Department
of Mathematics and Statistics, Queen's University, Kingston, Ontario,
Canada, K7L 3N6. Email: yuksel@mast.queensu.ca. Research is supported
in part by the Natural Sciences and Engineering Research Council of
Canada (NSERC) in the form of a Discovery Grant.%
}\renewcommand{\thefootnote}{\arabic{footnote}}
\setcounter{footnote}{0}
\end{abstract}

\section{\label{sec:Introduction}Introduction}

A linear system of dimension $n$ is said to be observable if an ensemble
of at most $n$ successive observations guarantee the recovery of
the initial state. Observability is an essential notion in control
theory as, with the sister notion of controllability, these form the
essence of modern linear control theory.

In this paper, we consider the observability problem when the number
of non-zeros in the initial state in a linear system is strictly less
than the dimension of the system. This might arise in systems where
natural or external forces give rise to a certain subset of components
of a linear system to be activated or excited, for example an external
force may give rise to a subset of locally unstable states while keeping
certain other states intact.

Furthermore, with the increasing emphasis on networked control systems,
it has been realized that the controllability and observability concepts
for linear systems with controllers having full access to sensory
information is not practical. Many research efforts have focused on
both stochastic settings, as well as information theoretic settings
to adapt the observability notion to control of linear systems with
limited information. One direction in this general field is the case
when the observations available at a controller comes at random intervals.
In this context, in both the information theory literature as well
as automatic control literature, a rich collection of papers have
studied the recursive estimation problem and its applications in remote
control \cite{LinearJump1,LinearJump2,LinearJump3,LinearJump4}.

In the following, we describe the system model. In Section \ref{sec:Preliminaries},
preliminaries on compressive sensing theory are presented. It follows
a formal discussion of observability of linear systems: since the
analytical tools and results are significantly different for different
cases, we first treat a deterministic setup in Section \ref{Sec:Deterministic}
and then study a stochastic setup in Section \ref{Sec:Stochastic}.
Detailed proofs are given in Section \ref{sec:Proofs}. Concluding
remarks are discussed in Section \ref{sec:Conclusion}.

\section{\label{sec:Problem-Formulation}Problem Formulation}

For the purpose of observability analysis, we consider the following
discrete-time linear time-invariant system (with zero control input):
$\bm{x}_{t+1}=\bm{Ax}_{t}$, $\bm{y}_{t}=\eta_{t}\bm{C}\bm{x}_{t}$,
where $t\in\mathbb{Z}_{+}$ denotes the discrete time instant, $\bm{x}_{t}\in\mathbb{R}^{n}$
and $\bm{y}_{t}\in\mathbb{R}^{d_{y}}$ are the state of the system
and the observation of the system respectively, the matrices $\bm{A}\in\mathbb{R}^{n\times n}$
and $\bm{C}\in\mathbb{R}^{d_{y}\times n}$ denote the state transfer
matrix and the observation matrix respectively, and $\eta_{t}$ takes
value either $0$ or $1$ ($\eta_{t}=1$ means an observation at time
$t$ is available, and $\eta=0$ otherwise).

The problem we are interested in is the observability of a system
with a sparse initial state: Given $m<n$ observations ($m$ instances
where $\eta_{t}=1$), can we reconstruct the initial state $\bm{x}_{0}\in\mathbb{R}^{n}$
exactly? Suppose that the receiver observes the output of the system
$\bm{y}_{t}$ at the (stopping) time instances $t_{1},t_{2},\cdots,t_{m}$.
Let the overall observation matrix be the stacked observation matrices
$\bm{O}_{\bm{T}_{m}}=\left[\left(\bm{C}\bm{A}^{t_{1}}\right)^{T},\left(\bm{C}\bm{A}^{t_{2}}\right)^{T},\cdots,\left(\bm{CA}^{t_{m}}\right)^{T}\right]^{T}$
and the overall observation be $\bm{y}_{\bm{T}_{m}}=\left[\bm{y}_{t_{1}}^{T},\bm{y}_{t_{2}}^{T},\cdots,\bm{y}_{t_{m}}^{T}\right]^{T},$
where the subscript $\bm{T}_{m}$ emphasizes that only the observations
at time instants $\bm{T}_{m}:=\{t_{1},t_{2},\cdots,t_{m}\}$ are available.
Then $\bm{y}_{\bm{T}_{m}}=\bm{O}_{\bm{T}_{m}}\bm{x}_{0}$. In order
to infer the initial state $\bm{x}_{0}$ from $\bm{y}_{\bm{T}_{m}}$,
the columns of $\bm{O}_{\bm{T}_{m}}$ have to be linearly independent,
or equivalently, the null-space of the matrix $\bm{O}_{\bm{T}_{m}}$
must be trivial.

While the general setup has been well understood, the problem of our
particular interest is the observability when the initial state $\bm{x}_{0}$
is sparse. The definition of a sparse vector is given as follows. 
\begin{defn}
Let $\bm{B}\in\mathbb{R}^{n\times n}$ be an orthonormal basis, i.e.,
$\bm{B}$ contains $n$ orthonormal columns. A vector $\bm{x}\in\mathbb{R}^{n}$
is $K$-sparse under $\bm{B}\in\mathbb{R}^{n\times n}$ if $\bm{x}=\bm{B}\bm{s}$
for some $\bm{s}\in\mathbb{R}^{n}$ with $\left\Vert \bm{s}\right\Vert _{0}\le K$,
where $||\bm{s}||_{0}$ gives the number of non-zero components in
the vector $\bm{s}$ ($\left\Vert \cdot\right\Vert _{0}$ is often
referred to as the $\ell_{0}$-norm, even though it is not a well-defined
norm). 
\end{defn}
\vspace{0cm}

Our formulation appears to be new in the control theory literature,
except for a paper \cite{WakinCDC2010} which considers a similar
setting for observability properties of a stochastic model to be considered
later in the paper. The differences between the approaches in the
stochastic setup are presented in Section \ref{Sec:Stochastic}. Another
related work is \cite{BasarACC2011} which designs control algorithms
based on sparsity in the state, where compressive sensing tools are
used to reconstruct the state for control purposes.

\section{\label{sec:Preliminaries}Preliminaries and Compressive Sensing }

Compressive sensing is a signal processing technique that encodes
a signal $\bm{x}$ of dimension $n$ by computing a measurement vector
$\bm{y}$ of dimension $m\ll n$ via linear projections, i.e., $\bm{y}=\bm{\Phi}\bm{x},$
where $\bm{\Phi}\in\mathbb{R}^{m\times n}$ is referred to as the
\emph{measurement matrix}. In general, it is not possible to uniquely
recover the unknown signal $\bm{x}$ using measurements $\bm{y}$
with reduced-dimensionality. Nevertheless, if the input signal is
sufficiently sparse, exact reconstruction is possible. In this context,
suppose that the unknown signal $\bm{x}\in\mathbb{R}^{n}$ is at most
$K$-sparse, i.e., that there are at most $K$ nonzero entries in
$\bm{x}$. A naive reconstruction method is to search among all possible
signals and find the sparsest one which is consistent with the linear
measurements. This method requires only $m=2K$ random linear measurements,
but finding the sparsest signal representation is an NP-hard problem.
On the other hand, Donoho and Candès et. al. \cite{Donoho_IT2006_CompressedSensing,Candes_Tao_IT2005_decoding_linear_programming}
demonstrated that reconstruction of $\bm{x}$ from $\bm{y}$ is a
\emph{polynomial time} problem if more measurements are taken. This
is achieved by casting the reconstruction problem as an \emph{$\ell_{1}$-minimization}
problem, i.e., $\min\;\left\Vert \bm{x}\right\Vert _{1}\;\mathrm{subject\; to}\;\bm{y}=\bm{\Phi}\bm{x},$
where $\left\Vert \bm{x}\right\Vert _{1}=\sum_{i=1}^{n}\left|x^{i}\right|$
denotes the $\ell_{1}$-norm of the vector $\bm{x}$. It is a convex
optimization problem and can be solved efficiently by linear programming
(LP) techniques. The reconstruction complexity equals $O\left(m^{2}n^{3/2}\right)$
if the convex optimization problem is solved using interior point
methods \cite{Nesterov_book1994_Interior_point_Convex_Programming}.
More recently, an iterative algorithm, termed \emph{subspace pursuit
(SP)}, was proposed independently in \cite{Dai_2008_Subspace_Pursuit}
and \cite{Tropp2008_CoSamp}. The corresponding computational complexity
is $O\left(Km(n+K^{2})\right)$, which is significantly smaller than
that of $\ell_{1}$-minimization when $K\ll n$.

A sufficient and necessary condition for $\ell_{1}$-minimization
to perform exact reconstruction is the so called \emph{the null-space
condition} \cite{Xu2010:CSGrassmannManifold}. 
\begin{thm}
If and only if for all $\bm{w}\in\mathbb{R}^{n}$ such that $\bm{\Phi}\bm{w}=\bm{0},$
and for all sets $T\subset\left\{ 1,2,\cdots,n\right\} $ such that
$\left|T\right|=K$, there exists a constant $c>1$ such that 
\begin{equation}
c\sum_{i\in T}\left|\bm{w}^{i}\right|\le\sum_{j\in T^{c}}\left|\bm{w}^{j}\right|,\label{eq:condition-xu}
\end{equation}
 where $T^{c}=\left\{ 1,2,\cdots,n\right\} -T$, then $\ell_{1}$-minimization
reconstructs $\bm{x}$ exactly. 
\end{thm}
\vspace{0cm}

A sufficient condition for both the $\ell_{1}$-minimization and SP
algorithms to perform exact reconstruction is based on the so called
\emph{restricted isometry property (RIP)} \cite{Candes_Tao_IT2005_decoding_linear_programming}.
A matrix $\bm{\Phi}\in\mathbb{R}^{m\times n}$ is said to satisfy
the Restricted Isometry Property (RIP) with coefficients $\left(K,\delta\right)$
for $K\le m$, $0\leq\delta\leq1$, if for all index sets $I\subset\left\{ 1,\cdots,n\right\} $
such that $\left|I\right|\le K$ and for all $\bm{q}\in\mathbb{R}^{\left|I\right|}$,
one has 
\[
\left(1-\delta\right)\left\Vert \bm{q}\right\Vert _{2}^{2}\le\left\Vert \bm{\Phi}_{I}\bm{q}\right\Vert _{2}^{2}\le\left(1+\delta\right)\left\Vert \bm{q}\right\Vert _{2}^{2},
\]
 where $\bm{\Phi}_{I}$ denotes the matrix formed by the columns of
$\bm{\Phi}$ with indices in $I$. The \emph{RIP parameter} $\delta_{K}$
is defined as the infimum of all parameters $\delta$ for which the
RIP holds. It was shown in \cite{Candes_Tao_IT2005_decoding_linear_programming,Candes_Tao_ApplMath2006_Stable_Signal_Recovery,Dai_2008_Subspace_Pursuit}
that both $\ell_{1}$-minimization and SP algorithms lead to exact
reconstructions of $K$-sparse signals if the matrix $\bm{\Phi}$
satisfies the RIP with a constant parameter, i.e., $\delta_{kK}\le c_{0}$
where $c_{0}\in\left(0,1\right)$ and $k\in\mathbb{R}_{+}$ are independent
of $K$. We note that different algorithms may have different parameter
values for $c_{0}$s and $k$s. Examples of random and deterministic
RIP matrices can be found in \cite{Candes_Tao_IT2006_Robust_Uncertainty_Principles,Candes_Tao_IT2005_decoding_linear_programming,RonaldA2007:CSMatrixConstruction,Mazumdar2011:RIPMatricesConstruction}.

For later use, we also consider a particular class of the measurement
matrices $\bm{\Phi}$. We will assume that $\bm{\Phi}^{T}\in\mathcal{S}_{n,m}\left(\mathbb{R}\right)$
(that is, the rows of $\bm{\Phi}\in\mathbb{R}^{m\times n}$ are orthonormal)
is isotropically distributed (the definition of $\mathcal{S}_{n,m}\left(\mathbb{R}\right)$
and the isotropic distribution on $\mathcal{S}_{n,m}\left(\mathbb{R}\right)$
will be introduced in Section \ref{sub:Stiefel-Manifold}). Under
this assumption, it has been shown in \cite{Rudelson2005_CS_error_correcting_codes}
that if the number of measurements satisfies $m\ge C\cdot K\log\left(n/K\right)$
for some positive constant $C$, then with high probability ($\ge1-e^{-nc}$
for some positive constant $c$) the $\ell_{1}$-minimization perfectly
reconstructs the input unknown signal $\bm{x}$.

\section{\label{Sec:Deterministic}The Deterministic Model}

This section characterizes the number of measurements needed for observability
for different scenarios. We assume that $\bm{x}_{0}$ is $K$-sparse
under a basis $\bm{B}\in\mathcal{S}_{n,n}\left(\mathbb{R}\right)$
and $\bm{B}$ is known in advance. Recall that observability generally
requires that the observability matrix $\bm{O}_{\bm{T}_{m}}$ has
full rank, i.e., at least $n$ measurements should be collected. When
$\bm{x}_{0}$ is sparse, the number of observations required for observability
can be significantly reduced.

We start with a special case where particular structures are imposed
on $\bm{A}$, $\bm{B}$ and $\bm{C}$ to reduce the number of required
observations to $2K+1$. 
\begin{prop}
\label{pro:observability-Vandermonde}Suppose that $\bm{x}_{0}$ is
$K$-sparse under the natural basis $\bm{B}=\bm{I}$. Assume that
$\bm{A}\in\mathbb{R}^{n\times n}$ is diagonal, and that all diagonal
entries are nonzero and distinct. Let all of the entries of $\bm{C}\in\mathbb{R}^{1\times n}$
($d_{y}=1$) be non-zero. Then $\bm{x}_{0}$ can be exactly reconstructed
after exactly $2K+1$ measurements by algorithms with polynomial complexity
in $n$. 
\end{prop}
\vspace{0in}

\begin{IEEEproof}
See Section \ref{sub:Proof-of-Pro-observability-Vandermonde}. \vspace{0in}
 \end{IEEEproof}
\begin{rem}
The reconstruction relies on the Reed-Solomon decoding method presented
in \cite{Tarokh2007_ISIT_reed_solomon_CS}. Note that the reconstruction
is not robust to noise and hence not very useful in practice. 
\end{rem}
\vspace{0in}

The following proposition considers the case where $\ell_{1}$-minimization
is used for reconstruction. We have further restrictions on the initial
state and observation time. 
\begin{prop}
\label{row}Let all of the entries of $\bm{C}\in\mathbb{R}^{1\times n}$
($d_{y}=1$) be non-zero. Suppose $c_{i}x_{0,i}\geq0$ for all $i$,
where $\bm{C}=\left[c_{1},\cdots,c_{n}\right]$. Further assume that
$\bm{A}\in\mathbb{R}^{n\times n}$ is diagonal, and that all diagonal
entries are nonzero. If the decoder receives $2K+1$ successive observations
at times $t=0,\dots,2K$, the decoder can reconstruct the initial
state perfectly and the unique solution can be obtained by the solution
of the linear program $\min||\bm{x}||_{1}\;\mbox{s.t. }\;\bm{O}_{\bm{t}}\bm{x}={\bf y},$
where $\bm{O}_{\bm{t}}=\left[\bm{C}^{T},\left(\bm{C}\bm{A}\right)^{T},\cdots,\left(\bm{CA}^{2K}\right)^{T}\right]^{T}.$ \end{prop}
\begin{IEEEproof}
See Section \ref{sub:ProofFuchs1}. 
\end{IEEEproof}
\vspace{0in}

We note that, one can relax the above to the case when the observations
are periodic such that $t_{2}-t_{1}=t_{3}-t_{2}=...=t_{m}-t_{m-1}$,
where $1,2,\dots,m$ are the observation times.

In the following, we consider more general settings. 
\begin{prop}
\label{RandomObservationCase} Suppose that $\bm{A}\in\mathbb{R}^{n\times n}$
is of Jordan canonical form, all diagonal entries are nonzero, and
the eigenvalues corresponding to different Jordan blocks are distinct.
Let the entries of $\bm{C}\in\mathbb{R}^{1\times n}$ ($d_{y}=1$)
be non-zero for all the leading components of Jordan blocks (that
is, for the first entry corresponding to a Jordan block). If the decoder
receives $m$ random observations, at random times $T_{m}=\{t_{1},t_{2},\dots,t_{m}\}$,
let $\bm{O}_{\bm{T}_{m}}=\left[\left(\bm{CA}^{t_{1}}\right)^{T},\left(\bm{CA}^{t_{2}}\right)^{T},\cdots,\left(\bm{CA}^{t_{m}}\right)^{T}\right]^{T}.$
Let $\bm{O}_{\bm{T}_{m}}(i)$ denote the $i^{th}$ column of $\bm{O}_{\bm{T}_{m}}$
for $1\leq i\leq n$. Define 
\[
M(\bm{T}_{m})=\sup_{i\neq j}\langle\frac{1}{||\bm{O}_{\bm{T}_{m}}(i)||_{2}}\bm{O}_{\bm{T}_{m}}(i),\frac{1}{||\bm{O}_{\bm{T}_{m}}(j)||_{2}}\bm{O}_{\bm{T}_{m}}(j)\rangle<1.
\]
 Then $\bm{x}_{0}$ can be exactly reconstructed after $m$ measurements
if: 
\[
\left\Vert \bm{x}_{0}\right\Vert _{0}\leq\frac{1}{2}(1+\frac{1}{M(\bm{T}_{m})})
\]
 by algorithms with polynomial complexity in $n$. In particular,
a linear program (LP) can be used to recover the initial state.\end{prop}
\begin{IEEEproof}
See Section \ref{sub:Proof-of-RandomObservationCase}. 
\end{IEEEproof}
\vspace{0in}

\begin{rem}
We recall that the observability of a linear system described by the
pair $(\bm{A},\bm{C})$ can be verified by the following criterion,
known as the Hautus-Rosenbrock test: The pair is observable if and
only if for all $\lambda\in\mathbb{C}$, the matrix $\left[\left(\lambda I-\bm{A}\right)^{T},\bm{C}^{T}\right]^{T}$
is full rank. Clearly, one needs to check the rank condition only
for the eigenvalues of $\bm{A}$. It is a consequence of the above
that, if the component of $\bm{C}$ corresponding to the first entry
of a Jordan block is zero, then the corresponding component cannot
be recovered even with $n$ successive observations, since this is
a necessary condition for observability. 
\end{rem}
\vspace{0cm}

A more general case is studied in the next proposition. 
\begin{prop}
\label{pro:observability-general-deterministic}Given $\bm{A}\in\mathbb{R}^{n\times n}$,
$\bm{C}\in\mathbb{R}^{d_{y}\times n}$ and $\bm{T}_{m}=\{t_{1},\cdots,t_{m}\}$,
if $\bm{\Phi}=\bm{O}_{\bm{T}_{m}}\bm{B}$ satisfies the null-space
condition (\ref{eq:condition-xu}), then $\ell_{1}$-minimization
$\min\;\left\Vert \bm{s}\right\Vert _{1}\;\mathrm{s.t.}\;\bm{y}_{\bm{t}}=\bm{O}_{\bm{T}_{m}}\bm{B}\bm{s}$
reconstructs $\bm{s}$ and $\bm{x}_{0}=\bm{B}\bm{s}$ exactly. Suppose
that $\bm{\Phi}$ satisfies the RIP with proper parameters, both $\ell_{1}$-minimization
and SP algorithm leads to exact reconstruction of the initial state
$\bm{x}_{0}$. 
\end{prop}
\vspace{0in}

This proposition is a direct application of the results presented
in Section \ref{sec:Preliminaries}. This result implies a protocol
in which one keeps collecting available observations $\bm{y}_{t_{1}},\bm{y}_{t_{2}},\cdots$
until the null-space or RIP condition is satisfied. However, the computation
complexity of verifying either of them generally increases exponentially
with $n$. There are two approaches to avoid this extremely expensive
computational cost. The first approach is reconstruction on the fly
by trying to reconstruct the unknown initial state $\bm{x}_{0}$ every
time when certain number of new observations are received; and continue
this process until the reconstruction is good enough. In the second
approach, certain suboptimal but computationally more efficient conditions,
for example, the incoherence condition, are employed to judge whether
current observations are sufficient for reconstruction.

\section{\label{Sec:Stochastic}The Stochastic Model}

In this section, we discuss a stochastic model for the system matrices.
One advantage of the stochastic model is that it helps in understanding
more general cases that are difficult to analyze using the deterministic
model. Examples include Theorem \ref{thm:observability-general} and
Corollary \ref{cor:observability-Gaussian}. Our analysis is based
on the concept of rotational invariance, defined in Subsection \ref{sub:Stiefel-Manifold}.
The intuition is that rotational invariance provides a rich structure
to {}``mix'' the non-zeros in the initial state and this {}``mixing''
ensures an observability with significantly reduced number of measurements.

During the preparation of this paper, we noticed that the stochastic
model was also discussed in an independent work \cite{WakinCDC2010}.
The major differences between our approach and that in \cite{WakinCDC2010}
are as follows. First, in \cite{WakinCDC2010}, the observation matrix
$\bm{C}_{k}$'s are assumed to be random Gaussian matrices. In contrast,
our model relies on rotationally invariant random matrices, which
are much more general. Second, though the work \cite{WakinCDC2010}
is targeted for general state transition matrix $\bm{A}$, the analysis
and results best suit for the $\bm{A}$ matrices with concentrated
spectrum, for example, unitary matrices. As a comparison, in our stochastic
model, we separate the rotational invariance and the spectral property
and hence the spectral property can be very much relaxed.

\subsection{\label{sub:Stiefel-Manifold}The Isotropy of Random Matrices}

To define rotational invariance, we need to define the set of rotational
matrices, often referred to as the Stiefel manifold. Formally, the
\emph{Stiefel manifold} $\mathcal{S}_{n,k}\left(\mathbb{R}\right)$
is defined as $\mathcal{S}_{n,k}\left(\mathbb{R}\right)=\left\{ \bm{U}\in\mathbb{R}^{n\times k}:\;\bm{U}^{T}\bm{U}=\bm{I}_{k}\right\} $,
where $\bm{I}_{k}$ is the $k\times k$ identity matrix. When $n=k$,
a matrix in $\mathcal{S}_{n,n}\left(\mathbb{R}\right)$ is an orthonormal
matrix and represents a rotation. A left rotation of a measurable
set $\mathcal{H}\subset\mathbb{R}^{m\times n}$ under a given rotation
represented by $\bm{A}\in\mathcal{S}_{m,m}$ is given by the set $\bm{A}\mathcal{H}=\left\{ \bm{A}\bm{H}:\;\bm{H}\in\mathcal{H}\right\} \subset\mathbb{R}^{n\times n}$.
Similarly defines the right rotation of $\mathcal{H}$ given by $\mathcal{H}\bm{B}$
for a given $\bm{B}\in\mathcal{S}_{n,n}$. An \emph{invariant}/\emph{isotropic}
probability measure $\mu_{I}$ \cite[Sections 2 and 3]{Muirhead_book82_multivariate_statistics,James1954_Normal_Multivariate_Analysis_Orthogonal_Group}
is defined by the property that for any measurable set $\mathcal{M}\subset\mathbb{R}^{m\times n}$
and rotation matrices $\bm{A}\in\mathcal{S}_{n,n}\left(\mathbb{R}\right)$
and $\bm{B}\in\mathcal{S}_{k,k}\left(\mathbb{R}\right)$, $\mu_{I}\left(\mathcal{M}\right)=\mu_{I}\left(\bm{A}\mathcal{M}\right)=\mu_{I}\left(\mathcal{M}\bm{B}\right).$
The invariant probability on the Stiefel manifold is essentially the
uniform probability measure, i.e., $\mu_{I}\left(\left\{ \bm{A}\in\mathcal{S}_{n,k}\left(\mathbb{R}\right):\;\left\Vert \bm{A}-\bm{U}\right\Vert _{F}\le\epsilon\right\} \right)$
is independent of the choice of $\bm{U}\in\mathcal{S}_{n,k}\left(\mathbb{R}\right)$.

The main results in this subsection are Lemmas \ref{lem:isotropy-stiefel}
and \ref{lem:isotropy-Gaussian-Jordan}, which show that an rotationally
invariant random matrix admits rotationally invariant matrix products
and decompositions. These results are the key for proving results
regarding observability in Subsection \ref{sub:Results-for-Stochastic-Models}. 
\begin{lem}
\label{lem:isotropy-stiefel} Let $\bm{A}\in\mathcal{S}_{n,k}\left(\mathbb{R}\right)$
be isotropically distributed. Let $\bm{B}\in\mathcal{S}_{n,n}\left(\mathbb{R}\right)$
be random. Let $\bm{C}=\bm{B}\cdot\bm{A}$. Then $\bm{C}\in\mathcal{S}_{n,k}\left(\mathbb{R}\right)$
is isotropically distributed and independent of $\bm{B}$. 
\end{lem}
\vspace{0in}

\begin{IEEEproof}
In order to show that $\bm{C}$ is independent of $\bm{B}$, it is
sufficient to show that for given arbitrary $\bm{B}\in\mathcal{S}_{n,n}\left(\mathbb{R}\right)$
and arbitrary measurable set $\mathcal{M}\subset\mathcal{S}_{n,k}\left(\mathbb{R}\right)$,
the conditional probability $\Pr\left(\bm{C}\in\mathcal{M}|\bm{B}\right)$
is independent of $\bm{B}$. This can be verified by observing 
\begin{align*}
\Pr\left(\bm{C}\in\mathcal{M}|\bm{B}\right)=\Pr\left(\bm{A}\in\bm{B}^{-1}\mathcal{M}|\bm{B}\right)\overset{\left(a\right)}{=}\Pr\left(\bm{A}\in\bm{B}^{-1}\mathcal{M}\right)\overset{\left(b\right)}{=}\Pr\left(\bm{A}\in\mathcal{M}\right)=\mu_{I}\left(\mathcal{M}\right),
\end{align*}
 where $\left(a\right)$ follows from the fact that $\bm{A}$ is independent
of $\bm{B}$, and $\left(b\right)$ comes from the facts that $\bm{A}$
is isotropically distributed and that $\bm{B}\in\mathcal{S}_{n,n}\left(\mathbb{R}\right)$
and hence $\bm{B}^{-1}=\bm{B}^{T}\in\mathcal{S}_{n,n}\left(\mathbb{R}\right)$.
This proves the lemma. 
\end{IEEEproof}
\vspace{0in}

Let $\bm{H}\in\mathbb{R}^{n\times n}$ be a standard Gaussian random
matrix, i.e., the entries of $\bm{H}$ are independent and identically
distributed Gaussian random variables with zero mean and unit variance.
Consider the Jordan matrix decomposition $\bm{H}=\bm{P}\bm{J}\bm{P}^{-1}$,
where $\bm{J}$ is often referred to as the Jordan normal form of
$\bm{H}$. Let $\bm{P}=\bm{U}_{\bm{P}}\bm{\Lambda}_{\bm{P}}\bm{V}_{\bm{P}}^{T}$
be the singular value decomposition of $\bm{P}$, where $\bm{\Lambda}_{\bm{P}}$
is the diagonal matrix composed of singular values of $\bm{P}$. Then
$\bm{P}^{-1}=\bm{V}_{\bm{P}}\bm{\Lambda}_{\bm{P}}^{-1}\bm{U}_{\bm{P}}^{T}$.
The following lemma states that the orthogonal matrix $\bm{U}_{\bm{P}}$
is isotropically distributed. 
\begin{lem}
\label{lem:isotropy-Gaussian-Jordan}Let $\bm{H}\in\mathbb{R}^{n\times n}$
be a standard Gaussian random matrix, let $\bm{H}=\bm{P}\bm{J}\bm{P}^{-1}$
be the corresponding Jordan matrix decomposition, and let $\bm{P}=\bm{U}_{\bm{P}}\bm{\Lambda}_{\bm{P}}\bm{V}_{\bm{P}}^{T}$
be the singular value decomposition of $\bm{P}$. Then $\bm{U}_{\bm{P}}\in\mathcal{S}_{n,n}\left(\mathbb{R}\right)$
is isotropically distributed and independent of $\bm{J}$, $\bm{\Lambda}_{\bm{P}}$
and $\bm{V}_{\bm{P}}$. 
\end{lem}
\vspace{0.02in}

\begin{IEEEproof}
According to the statement of this lemma, $\bm{H}$ is a standard
Gaussian random matrix. Hence, the distribution of $\bm{H}$ is left
and right rotationally invariant \cite[pg. 37]{James1964_distributions_random_matrices,Edelman1989_phd_thesis}.
That is, for measurable sets $\mathcal{H}\subset\mathbb{R}^{n\times n}$
and arbitrary $\bm{Q}\in\mathcal{S}_{n,n}\left(\mathbb{R}\right)$,
$\Pr\left(\bm{H}\in\mathcal{H}\right)=\Pr\left(\bm{H}\in\bm{Q}\mathcal{H}\right)=\Pr\left(\bm{H}\in\mathcal{H}\bm{Q}\right),$
and therefore, $\Pr\left(\bm{H}\in\mathcal{H}\right)=\Pr\left(\bm{H}\in\bm{Q}\mathcal{H}\bm{Q}^{T}\right).$
To simplify the notation, let $\bm{H}=\bm{U}_{\bm{P}}\bm{B}\bm{U}_{\bm{P}}^{T}$,
where $\bm{B}=\bm{\Lambda}_{\bm{P}}\bm{V}_{\bm{P}}^{T}\bm{J}\bm{V}_{\bm{P}}\bm{\Lambda}_{\bm{P}}^{-1}$.
Let $\mathcal{U}\subset\mathcal{S}_{n,n}\left(\mathbb{R}\right)$
be an arbitrary measurable set of $\bm{U}_{\bm{P}}$ . Let $\Pr\left(\mathcal{U}\right)$
be the probability measure of $\bm{U}_{\bm{P}}$ induced from the
probability measure of $\bm{H}$.

The isotropics of $\bm{U}_{\bm{P}}$ means that $\Pr\left(\bm{U}_{\bm{P}}\in\mathcal{U}\right)=\Pr\left(\bm{U}_{\bm{P}}\in\bm{Q}\mathcal{U}\right)$
for an arbitrarily given $\bm{Q}\in\mathcal{S}_{n,n}\left(\mathbb{R}\right)$.
To reach this end, note that $\Pr\left(\bm{U}_{\bm{P}}\in\mathcal{U}\right)=\Pr\left\{ \bm{H}:\;\exists\bm{U}_{\bm{P}}\in\mathcal{U}\;\mbox{s.t.}\;\bm{H}=\bm{U}_{\bm{P}}\bm{B}\bm{U}_{\bm{P}}^{T}\right\} ,$
and 
\begin{align*}
 & \Pr\left(\bm{U}_{\bm{P}}^{\prime}\in\bm{Q}\mathcal{U}\right)=\Pr\left\{ \bm{H}^{\prime}:\;\exists\bm{U}_{\bm{P}}^{\prime}\in\bm{Q}\mathcal{U}\;\mbox{s.t.}\;\bm{H}^{\prime}=\bm{U}_{\bm{P}}^{\prime}\bm{B}\bm{U}_{\bm{P}}^{\prime T}\right\} \\
 & =\Pr\left\{ \bm{H}^{\prime}:\;\exists\bm{U}_{\bm{P}}\in\mathcal{U}\;\mbox{s.t.}\;\bm{H}^{\prime}=\bm{Q}\left(\bm{U}_{\bm{P}}\bm{B}\bm{U}_{\bm{P}}^{T}\right)\bm{Q}^{T}\right\} .
\end{align*}
 In other words, for any $\bm{H}$ that induces a $\bm{U}_{\bm{P}}\in\mathcal{U}$,
$\bm{Q}\bm{H}\bm{Q}^{T}$ induces a $\bm{U}_{\bm{P}}\in\bm{Q}\mathcal{U}$,
and vice versa. Because we have shown $\Pr\left(\bm{H}\in\mathcal{H}\right)=\Pr\left(\bm{H}\in\bm{Q}\mathcal{H}\bm{Q}^{T}\right)$,
we conclude that $\bm{U}_{\bm{P}}$ is isotropically distributed.
Furthermore, the above argument also suggests that $\bm{U}_{\bm{P}}$
is independent of the matrix $\bm{B}$, therefore independent of $\bm{J}$,
$\bm{\Lambda}_{\bm{P}}$ and $\bm{V}_{\bm{P}}$. This lemma is proved. \end{IEEEproof}
\begin{rem}
\label{rem:isotropy-Jordan-general}Although Lemma \ref{lem:isotropy-Gaussian-Jordan}
only treats standard Gaussian random matrices, the same result holds
for general random matrix ensembles whose distributions are left and
right rotationally invariant: The proof of Lemma \ref{lem:isotropy-Gaussian-Jordan}
can be carried over. 
\end{rem}

\subsection{\label{sub:Results-for-Stochastic-Models}Results for Stochastic
Models}

Recall that a general linear system is observable if and only if the
observability matrix $\bm{O}_{\bm{T}_{m}}$ has full row rank. One
may expect that the row rank of $\bm{O}_{\bm{T}_{m}}$ still indicates
the observability of a linear system with sparse initial state and
partial observations. The next theorem confirms the intimate relation
between the row rank and the observability. The difference between
our results and the standard results is that the required minimum
rank is much smaller than the signal dimension $n$ in our setting. 
\begin{thm}
\label{thm:observability-general} Suppose that $\bm{A}\in\mathbb{R}^{n\times n}$
and $\bm{C}\in\mathbb{R}^{d_{y}\times n}$ are independent drawn from
a random matrix ensemble whose distribution is left and right rotationally
invariant. Let $r$ be the row rank of the overall observation matrix
$\bm{O}_{\bm{T}_{m}}$. If $r\ge O\left(K\log\frac{n}{K}\right)$,
then the $\ell_{1}$-minimization method perfectly reconstructs $\bm{x}_{0}$
from $\bm{y}_{\bm{t}}=\bm{O}_{\bm{t}}\bm{x}_{0}$ (where we write
$\bm{t}=\bm{T}_{m}$ for notational convenience) with high probability
(at least $1-e^{-nc}$ for some positive constant $c$ independent
of $n$ and $r$). 
\end{thm}
\vspace{0.01in}

The proof of Theorem \ref{thm:observability-general} rests on the
following Lemma. 
\begin{lem}
\label{lem:observability-general} Assume the same set-ups as in Theorem
\ref{thm:observability-general} and let $\bm{t}=\bm{T}_{m}$ for
notational convenience. Let $\bm{O}_{\bm{t}}=\bm{U}_{\bm{t}}\bm{\Lambda}_{\bm{t}}\bm{V}_{\bm{t}}^{T}$
be the corresponding singular value decomposition, where $\bm{U}_{\bm{t}}\in\mathcal{S}_{md_{y},md_{y}}\left(\mathbb{R}\right)$,
$\bm{V}_{\bm{t}}\in\mathcal{S}_{n,n}\left(\mathbb{R}\right)$ are
the left and right singular vector matrices respectively. Then $\bm{V}_{\bm{t}}$
is isotropically distributed and independent of $\bm{U}_{\bm{t}}$
and $\bm{\Lambda}_{\bm{t}}$. 
\end{lem}
\vspace{0.01in}

While Lemma \ref{lem:observability-general} is proved in Section
\ref{sub:Proof-of-Lem-observability-general}, the detailed proof
of Theorem \ref{thm:observability-general} is presented in Section
\ref{sub:Proof-of-Thm-observability-general}. The detailed reconstruction
procedure using $\ell_{1}$-minimization is explicitly presented in
the proof.

The next corollary presents a special case where the diagonal form
is involved. 
\begin{cor}
\label{cor:observability-Gaussian}Suppose that $\bm{A}\in\mathbb{R}^{n\times n}$
and $\bm{C}\in\mathbb{R}^{1\times n}$ ($d_{y}=1$) are independent
drawn from random matrix ensembles whose distribution is left and
right rotationally invariant. Suppose that the Jordan normal form
$\bm{J}=\bm{P}^{-1}\bm{A}\bm{P}$ is diagonal with distinct diagonal
entries with probability one. Then after $m\ge O\left(K\log\frac{n}{K}\right)$
measurements, the $\ell_{1}$-minimization method perfectly reconstructs
$\bm{x}_{0}$ with high probability (at least $1-e^{-nc}$ for some
positive constant $c$). 
\end{cor}
\begin{proof} See Section \ref{sub:Proof-of-Cor-observability-Gaussian}.
\end{proof} \vspace{0.01in}

Acute readers may ask whether there exists a random matrix ensemble
such that the random sample $\bm{A}$ satisfies the required conditions
in Corollary \ref{cor:observability-Gaussian}. In fact, if $\bm{A}=\bm{H}\bm{H}^{T}$
where $\bm{H}\in\mathbb{R}^{n\times n}$ is a standard Gaussian random
matrix, then all the conditions required for $\bm{A}$ hold. This
corollary guarantees that blindly collecting $m\ge O\left(K\log\frac{n}{K}\right)$
observations is sufficient for perfect reconstruction with high probability.

\section{\label{sec:Proofs}Proofs}

\subsection{\label{sub:Proof-of-Pro-observability-Vandermonde}Proof of Proposition
\ref{pro:observability-Vandermonde}}

Let $\bm{A}=\mathrm{diag}\left(\bm{\lambda}\right)$ where $\bm{\lambda}=\left[\lambda_{1},\lambda_{2},\cdots,\lambda_{n}\right]^{T}$
is the vector containing the diagonal entries of $\bm{A}$. Let $c_{i}$
denote the $i^{\mathrm{th}}$ entry of the row vector $\bm{C}$. Then,
$\bm{C}\bm{A}^{t_{i}}=\left[c_{1}\lambda_{1}^{t_{i}},c_{2}\lambda_{2}^{t_{i}},\cdots,c_{n}\lambda_{n}^{t_{i}}\right]=\left[\lambda_{1}^{t_{i}},\lambda_{2}^{t_{i}},\cdots,\lambda_{n}^{t_{i}}\right]\mathrm{diag}\left(\bm{C}\right),$
where $\mathrm{diag}\left(\bm{C}\right)$ is the diagonal matrix whose
$i^{\mathrm{th}}$ diagonal entry is $c_{i}$. Hence, 
\begin{align*}
\bm{y}_{\bm{T}_{m}}=\bm{O}_{\bm{T}_{m}}\bm{x}_{0}=\underbrace{\left[\begin{array}{cccc}
\lambda_{1}^{t_{1}} & \lambda_{2}^{t_{1}} & \cdots & \lambda_{n}^{t_{1}}\\
\lambda_{1}^{t_{2}} & \lambda_{2}^{t_{2}} & \cdots & \lambda_{n}^{t_{2}}\\
\vdots & \vdots & \ddots & \vdots\\
\lambda_{1}^{t_{m}} & \lambda_{2}^{t_{m}} & \cdots & \lambda_{n}^{t_{m}}
\end{array}\right]}_{\bm{\Lambda}_{\bm{t}}}\mathrm{diag}\left(\bm{C}\right)\bm{x}_{0}.
\end{align*}
 Since all the entries of $\bm{C}$ are non-zero, $\mathrm{diag}\left(\bm{C}\right)\bm{x}_{0}$
is $K$-sparse under the natural basis. On the other hand, since $\lambda_{1},\lambda_{2},\cdots,\lambda_{n}$
are all distinct, the matrix $\bm{\Lambda}_{\bm{t}}$ is a truncation
of the full rank Vandermonde matrix \cite{Horn1991_book_matrix_analysis}.
Now according to the Reed-Solomon decoding method presented in \cite{Tarokh2007_ISIT_reed_solomon_CS}
and the corresponding proof, as long as $m\ge2K+1$, one can exactly
reconstruct $\mathrm{diag}\left(\bm{C}\right)\bm{x}_{0}$ and therefore
$\bm{x}_{0}$ from $\bm{y}_{\bm{t}}$ with the number of algebraic
operations polynomial in $n$. This proposition is therefore proved.

\subsection{\label{sub:ProofFuchs1}Proof of Proposition \ref{row}}

We first consider the case when $\bm{A}$ is diagonal. Since $\bm{A}$
is diagonal, it is of the form $\bm{A}=\mbox{diag}\left(\left[\lambda_{1},\cdots,\lambda_{n}\right]\right)$.
Furthermore, assume that $\bm{C}=\left[c_{1},\cdots,c_{n}\right]$
is a row vector. With $m$ many successive observations, we have a
linear system described by 
\begin{align*}
\bm{y}_{\bm{t}} & =\underbrace{\left[\begin{array}{cccc}
1 & 1 & \cdots & 1\\
\lambda_{1} & \lambda_{2} & \cdots & \lambda_{n}\\
\vdots & \vdots & \ddots & \vdots\\
\lambda_{1}^{m-1} & \lambda_{2}^{m-1} & \cdots & \lambda_{n}^{m-1}
\end{array}\right]}_{\bm{M}}\mathrm{diag}\left(\left[c_{1},\cdots,c_{n}\right]\right)\bm{x}_{0}.
\end{align*}
 Define $\bm{z}\in\mathbb{R}^{n}$ such that $z_{i}=c_{i}x_{0,i}\ge0$.
Then the corresponding $\ell_{1}$-minimization problem becomes 
\begin{equation}
\underset{\bm{z}}{\min}\;\left\Vert \bm{z}\right\Vert _{1}\;\mbox{subject to }\bm{y}_{t}=\bm{M}\bm{z}.\label{eq:L1-Z}
\end{equation}
 Once we solve the above optimization prolem, it is clear that $x_{0,i}=z_{i}/\left(\lambda_{i}^{t_{1}}c_{i}\right)$
where $t_{1}=0$.

For this case, we first show that the $\ell_{1}$-minimization has
a unique solution. Via duality theory, for a constrained minimization
problem of a convex function with an equality constraint, the minimization
has a unique solution if one can find a Lagrange multiplier (in the
dual space) for which the Lagrangian at the solution is locally stationary.
More specifically, let $\bm{M}_{:,i}$ be the $i^{\mbox{th}}$ column
of the matrix $\bm{M}$. Let $i_{1},\cdots,i_{K}$ be the indices
of the nonzero entries of $\bm{x}_{0}$. Clearly, $i_{1},\cdots,i_{K}$
are also the indices of the nonzero entries of the corresponding $\bm{z}=\mathrm{diag}\left(\left[\lambda_{1}^{t_{1}}c_{1},\cdots,\lambda_{n}^{t_{1}}c_{n}\right]\right)\bm{x}_{0}$.
If there exists a vector $\bm{g}\in\mathbb{R}^{m}$ so that 
\[
\begin{cases}
\left\langle \bm{g},\bm{M}_{:,i}\right\rangle =1 & \;\forall i\in\left\{ i_{1},i_{2},\cdots,i_{K}\right\} \\
\left\langle \bm{g},\bm{M}_{:,i}\right\rangle <1 & \;\forall i\notin\left\{ i_{1},i_{2},\cdots,i_{K}\right\} 
\end{cases},
\]
 then the duality theory implies that the optimization problem in
(\ref{eq:L1-Z}) has a unique minimizer that is $K$-sparse and has
nonzero entries at indices $i_{1},\cdots,i_{K}$.

In the following we construct a subdifferential which is essentially
what Fuchs constructed in \cite{Fuchs2005_ICASSP_sparsity}. Consider
a polynomial in $\lambda$ of the form $P\left(\lambda\right)=\prod_{k=1}^{K}(\lambda_{i_{k}}-\lambda)^{2}=\alpha_{0}\lambda^{2K}+\alpha_{1}\lambda^{2K-1}+\dots+\alpha_{2K}.$
It is clear that 
\[
\begin{cases}
P\left(\lambda_{i}\right)=0 & \;\forall i\in\left\{ i_{1},i_{2},\cdots,i_{K}\right\} \\
P\left(\lambda_{i}\right)>0 & \;\forall i\notin\left\{ i_{1},i_{2},\cdots,i_{K}\right\} 
\end{cases},
\]
 where the inequality holds since $\lambda_{i}$'s are distinct. Let
$\bm{f}\in\mathbb{R}^{m},\bm{f}:=\left[\alpha_{2K},\alpha_{2K-1},\cdots,\alpha_{1},\alpha_{0},0,0,\cdots,0\right]^{T}.$
It can be verified that the inner product $\left\langle \bm{f},\left[1,\lambda_{i},\cdots,\lambda_{i}^{m-1}\right]^{T}\right\rangle =\prod_{k=1}^{K}(\lambda_{i_{k}}-\lambda_{i})^{2}=P\left(\lambda_{i}\right).$
Now, define a vector $\bm{g}\in\mathbb{R}^{m}$ as $\bm{g}=\left[1,0,0,\cdots,0\right]^{T}-\bm{f}$.
Then $\left\langle \bm{g},\left[1,\lambda_{i},\cdots,\lambda_{i}^{m-1}\right]^{T}\right\rangle =1-P\left(\lambda_{i}\right)$.
The vector $\bm{g}$ is the desired Lagrange vector. Hence, the optimization
problem (\ref{eq:L1-Z}) has a unique minimizer.

What now needs to be shown is that there is a unique solution to the
original problem under the $l_{0}$ constraint. In other words, we
wish to show that there is a unique $K-$sparse $\bm{z}$ such that
$\bm{y}_{t}=\bm{M}\bm{z}$. Now, let there be another $K-$sparse
solution $\bm{z}^{\prime}$. Then, $\bm{M}(\bm{z}-\bm{z}^{\prime})=\bm{0}$.
But, since any $2K$ columns of the Vandermonde matrix $\bm{M}$ are
linearly independent, $\bm{z}-\bm{z}^{\prime}$ has to be the zero
vector. Hence, this ensures the the found $\ell_{1}$ solution is
the sought $l_{0}$ solution. \hfill{}$\diamond$

\subsection{\label{sub:Proof-of-RandomObservationCase}Proof of Proposition \ref{RandomObservationCase}}

We now discuss the result for a Jordan matrix ${\bf A}$. Observe
that 
\[
J=\begin{bmatrix}\lambda_{1} & 1 & 0\\
0 & \lambda_{1} & 1\\
0 & 0 & \lambda_{1}
\end{bmatrix}\quad\Rightarrow\quad J^{n}=\begin{bmatrix}\lambda_{1}^{n} & {n \choose 1}\lambda_{1}^{n-1} & {n \choose 2}\lambda_{1}^{n-2}\\
0 & \lambda_{1}^{n} & {n \choose 1}\lambda_{1}^{n-1}\\
0 & 0 & \lambda_{1}^{n}
\end{bmatrix}.
\]
 Thus, it follows that if $\bm{A}$ is of the diagonal form: $diag(\lambda_{1},\dots,\lambda_{n})$,
 the random observation matrix writes as: 
\[
M=\begin{bmatrix}c_{1}\lambda_{1}^{t_{1}} & c_{1}t_{1}\lambda_{1}^{t_{1}-1}+c_{2}\lambda_{1}^{t_{1}} & \cdots & c_{n}\lambda_{n}^{t_{1}}\\
\vdots & \vdots & \ddots & \vdots\\
c_{1}\lambda_{1}^{t_{m}} & c_{1}t_{m}\lambda_{1}^{t_{m}-1}+c_{2}\lambda_{1}^{t_{m}} & \cdots & c_{n}\lambda_{n}^{t_{m}}
\end{bmatrix}
\]

If $c_{1}$ is non-zero, and the entries corresponding to leading
entries of Jordan blocks are non-zero, the columns of the matrix become
linearly independent. By multiplying the initial condition with a
diagonal matrix, we can normalize the columns such that the $l_{2}$
norm of each column is equal to $1$.

The rest of the proof now follows from Theorem 3 of \cite{FuchsIT2004}.\hfill{}$\diamond$

\subsection{\label{sub:Proof-of-Lem-observability-general}Proof of Lemma \ref{lem:observability-general}}

Consider the Jordan decomposition $\bm{A}=\bm{P}\bm{J}\bm{P}^{-1}$
and the singular value decomposition $\bm{P}=\bm{U}_{\bm{P}}\bm{\Lambda}_{\bm{P}}\bm{V}_{\bm{P}}^{T}$.
It is clear that $\bm{P}^{-1}=\bm{V}_{\bm{P}}\bm{\Lambda}_{\bm{P}}^{-1}\bm{U}_{\bm{P}}^{T}$.
For notational compactness, let $\tilde{\bm{A}}=\bm{\Lambda}_{\bm{P}}\bm{V}_{\bm{P}}^{T}\bm{J}\bm{V}_{\bm{P}}\bm{\Lambda}_{\bm{P}}^{-1}$
so that $\bm{A}=\bm{U}_{\bm{P}}\tilde{\bm{A}}\bm{U}_{\bm{P}}^{T}$.
It is elementary to verify that $\bm{A}^{t_{i}}=\bm{U}_{\bm{P}}\tilde{\bm{A}}^{t_{i}}\bm{U}_{\bm{P}}^{T}$.
Hence, 
\[
\bm{O}_{\bm{t}}=\left[\begin{array}{c}
\bm{C}\bm{A}^{t_{1}}\\
\vdots\\
\bm{C}\bm{A}^{t_{m}}
\end{array}\right]=\left[\begin{array}{c}
\bm{C}\bm{U}_{\bm{P}}\tilde{\bm{A}}^{t_{1}}\bm{U}_{\bm{P}}^{T}\\
\vdots\\
\bm{C}\bm{U}_{\bm{P}}\tilde{\bm{A}}^{t_{m}}\bm{U}_{\bm{P}}^{T}
\end{array}\right].
\]

We shall show that $\bm{U}_{\bm{P}}$ is independent of both $\tilde{\bm{A}}$
and $\bm{C}\bm{U}_{\bm{P}}$. Since $\bm{A}$ is left and right rotation-invariantly
distributed, according to Remark \ref{rem:isotropy-Jordan-general},
$\bm{U}_{\bm{P}}$ is isotropically distributed and independent of
$\tilde{\bm{A}}$. In order to show that $\bm{U}_{\bm{P}}$ is independent
of $\bm{C}\bm{U}_{\bm{P}}$, we resort to the singular value decomposition
$\bm{C}=\bm{U}_{\bm{C}}\bm{\Lambda}_{\bm{C}}\bm{V}_{\bm{C}}^{T}$.
Since $\bm{C}$ is right rotation-invariantly distributed, $\bm{V}_{\bm{C}}$
is isotropically distributed. Thus $\tilde{\bm{V}}_{\bm{C}}^{T}:=\bm{V}_{\bm{C}}^{T}\bm{U}_{\bm{P}}$
is isotropically distributed and independent of $\bm{U}_{\bm{P}}$
according to Lemma \ref{lem:isotropy-stiefel}. As a result, $\bm{C}\bm{U}_{\bm{P}}=\bm{U}_{\bm{C}}\bm{\Lambda}_{\bm{C}}\tilde{\bm{V}}_{\bm{C}}^{T}$
is independent of $\bm{U}_{\bm{P}}$. Write $\bm{O}_{\bm{t}}=\tilde{\bm{O}}_{\bm{t}}\bm{U}_{\bm{P}}^{T}$,
where $\tilde{\bm{O}}_{\bm{t}}=\left[\left(\bm{C}\bm{U}_{\bm{P}}\tilde{\bm{A}}^{t_{1}}\right)^{T},\cdots,\left(\bm{C}\bm{U}_{\bm{P}}\tilde{\bm{A}}^{t_{m}}\right)^{T}\right]^{T}.$
Since $\bm{U}_{\bm{P}}$ is independent of both $\tilde{\bm{A}}$
and $\bm{C}\bm{U}_{\bm{P}}$, $\bm{U}_{\bm{P}}$ is independent of
$\tilde{\bm{O}}_{\bm{t}}$. Write the singular value decompositions
of $\bm{O}_{\bm{t}}$ and $\tilde{\bm{O}}_{\bm{t}}$ as $\bm{O}_{\bm{t}}=\bm{U}_{\bm{t}}\bm{\Lambda}_{\bm{t}}\bm{V}_{\bm{t}}^{T}$
and $\tilde{\bm{O}}_{\bm{t}}=\bm{U}_{\bm{t}}\bm{\Lambda}_{\bm{t}}\tilde{\bm{V}}_{\bm{t}}^{T}$.
Clearly $\bm{V}_{\bm{t}}=\bm{U}_{\bm{P}}\tilde{\bm{V}}_{\bm{t}}$.
Since $\bm{U}_{\bm{P}}$ is isotropically distributed and independent
of $\tilde{\bm{O}}_{\bm{t}}$, $\bm{V}_{\bm{t}}=\bm{U}_{\bm{P}}\tilde{\bm{V}}_{\bm{t}}$
is isotropically distributed and independent of both $\bm{\Lambda}_{\bm{t}}$
and $\bm{U}_{\bm{t}}$ according to Lemma \ref{lem:isotropy-stiefel}.
This completes the proof.

\subsection{\label{sub:Proof-of-Thm-observability-general}Proof of Theorem \ref{thm:observability-general}}

We transfer the considered reconstruction problem to the standard
compressive sensing reconstruction. Let $\lambda_{1},\lambda_{2},\cdots,\lambda_{r}$
be the $r$ non-zero singular values of $\bm{O}_{\bm{t}}$ and $\bm{\lambda}=\left[\lambda_{1},\lambda_{2},\cdots,\lambda_{r}\right]^{T}$.
The singular value decomposition of $\bm{O}_{\bm{t}}$ can be written
in the form 
\[
\bm{O}_{\bm{t}}=\bm{U}_{\bm{t}}\left[\begin{array}{cc}
\mathrm{diag}\left(\bm{\lambda}\right) & \bm{0}\\
\bm{0} & \bm{0}
\end{array}\right]\bm{V}_{\bm{t}}^{T},
\]
 where $\mathrm{diag}\left(\bm{\lambda}\right)$ is the diagonal matrix
generated from $\bm{\lambda}$. Note that 
\[
\bm{U}_{\bm{t}}^{T}\bm{y}_{\bm{t}}=\left[\begin{array}{cc}
\mathrm{diag}\left(\bm{\lambda}\right) & \bm{0}\\
\bm{0} & \bm{0}
\end{array}\right]\bm{V}_{\bm{t}}^{T}\bm{x}_{0}.
\]
 The $r+1,r+2,\cdots,m$ entries of $\bm{U}_{\bm{t}}^{T}\bm{y}_{\bm{t}}$
are zeros: they do not carry any information about $\bm{x}_{0}$.
Define $\tilde{\bm{y}}_{\bm{t}}$ be the vector containing the first
$r$ entries of $\bm{U}_{\bm{t}}^{T}\bm{y}_{\bm{t}}$. We have $\tilde{\bm{y}}_{\bm{t}}=\left[\begin{array}{cc}
\mathrm{diag}\left(\bm{\lambda}\right) & \bm{0}\end{array}\right]\bm{V}_{\bm{t}}^{T}\bm{x}_{0}$ and therefore 
\begin{align}
\mathrm{diag}\left(\bm{\lambda}\right)^{-1}\tilde{\bm{y}}_{\bm{t}} & =\left[\begin{array}{cc}
\bm{I}_{r} & \bm{0}\end{array}\right]\bm{V}_{\bm{t}}^{T}\bm{x}_{0}=\left[\begin{array}{cc}
\bm{I}_{r} & \bm{0}\end{array}\right]\bm{V}_{\bm{t}}^{T}\bm{B}\bm{s},\label{eq:T1-01}
\end{align}
 where $\bm{I}_{r}$ is the $r\times r$ identity matrix.

The unknown $\bm{s}$ ($K$-sparse) can be reconstructed by $\ell_{1}$-minimization
with high probability. Since $\bm{V}_{t}$ is isotropically distributed
and independent of $\bm{B}$, the matrix $\bm{V}_{\bm{t}}^{T}\bm{B}$
is isotropically distributed. The matrix $\left(\left[\begin{array}{cc}
\bm{I}_{r} & \bm{0}\end{array}\right]\bm{V}_{\bm{t}}^{T}\bm{B}\right)^{T}\in\mathcal{S}_{n,r}\left(\mathbb{R}\right)$, containing the first $r$ rows of $\bm{V}_{\bm{t}}^{T}\bm{B}$ as
columns, is therefore isotropically distributed. Provided that $r\ge O\left(K\log\left(n/K\right)\right)$,
the unknown signal $\bm{s}$ can be exactly reconstructed from $\mathrm{diag}\left(\bm{\lambda}\right)^{-1}\tilde{\bm{y}}_{\bm{t}}$
via $\ell_{1}$-minimization \cite{Rudelson2005_CS_error_correcting_codes}.
Theorem \ref{thm:observability-general} is proved. 
\begin{rem}
The reconstruction procedure involves singular value decomposition,
matrix production, and $\ell_{1}$-minimization. The numbers of algebraic
operations required for all these steps are polynomial in $n$. Hence,
the complexity of the whole reconstruction process is polynomial in
$n$. 
\end{rem}

\subsection{\label{sub:Proof-of-Cor-observability-Gaussian}Proof of Corollary
\ref{cor:observability-Gaussian}}

Since both $\bm{A}$ and $\bm{C}$ are left and right rotation-invariantly
distributed, Theorem \ref{thm:observability-general} can be applied.
Let $\bm{A}=\bm{P}\bm{J}\bm{P}^{^{-1}}$ be a Jordan decomposition.
Corollary \ref{cor:observability-Gaussian} holds if 
\[
\bm{O}_{\bm{t}}=\left[\begin{array}{c}
\bm{C}\bm{A}^{t_{1}}\\
\bm{C}\bm{A}^{t_{2}}\\
\vdots\\
\bm{C}\bm{A}^{t_{m}}
\end{array}\right]=\left[\begin{array}{c}
\bm{C}\bm{P}\bm{J}^{t_{1}}\\
\bm{C}\bm{P}\bm{J}^{t_{2}}\\
\vdots\\
\bm{C}\bm{P}\bm{J}^{t_{m}}
\end{array}\right]\bm{P}^{-1}
\]
 is full row ranked with probability one, i.e., $\mathrm{rank}\left(\bm{O}_{\bm{t}}\right)=m\ge O\left(K\log\frac{n}{K}\right)$
with probability one.

Suppose that the Jordan normal form $\bm{J}=\bm{P}^{-1}\bm{A}\bm{P}$
is diagonal. Denote the $j^{\mathrm{th}}$ diagonal entry of $\bm{J}$
by $J_{i}$. Note that 
\begin{align*}
\bm{CP}\bm{J}^{t_{i}} & =\left[\left(\bm{CP}\right)_{1}J_{1}^{t_{i}},\left(\bm{CP}\right)_{2}J_{2}^{t_{i}},\cdots,\left(\bm{CP}\right)_{n}J_{n}^{t_{i}}\right]\\
 & =\left[J_{1}^{t_{i}},J_{2}^{t_{i}},\cdots,J_{n}^{t_{i}}\right]\mathrm{diag}\left(\bm{CP}\right),
\end{align*}
 where $\mathrm{diag}\left(\bm{CP}\right)$ is the diagonal matrix
generated from the row vector $\bm{CP}$. Define 
\[
\bm{J}_{V,\bm{t}}=\left[\begin{array}{cccc}
J_{1}^{t_{1}} & J_{2}^{t_{1}} & \cdots & J_{n}^{t_{1}}\\
J_{1}^{t_{2}} & J_{2}^{t_{2}} & \cdots & J_{n}^{t_{2}}\\
\vdots & \vdots & \ddots & \vdots\\
J_{1}^{t_{m}} & J_{2}^{t_{m}} & \cdots & J_{n}^{t_{m}}
\end{array}\right].
\]
 Then $\bm{O}_{\bm{t}}=\bm{J}_{V,\bm{t}}\mathrm{diag}\left(\bm{CP}\right)\bm{P}^{-1}.$
Note that $\bm{J}_{V}$ is composed of $m$ rows of the Vandemonde
matrix 
\[
\bm{J}_{V}=\left[\begin{array}{cccc}
1 & 1 & \cdots & 1\\
J_{1} & J_{2} & \cdots & J_{n}\\
\vdots & \vdots & \ddots & \vdots\\
J_{1}^{m-1} & J_{2}^{m-1} & \cdots & J_{n}^{m-1}
\end{array}\right].
\]
 The matrix $\bm{J}_{V,\bm{t}}$ has full row rank. By definition
of $\bm{P}$, $\bm{P}^{-1}$ has full rank as well. Therefore, $\bm{O}_{t}$
has full row rank if and only if $\bm{CP}$ does not contain any zero
entries.

The fact that the row vector $\bm{CP}$ does not contain any zero
entries holds with probability one. This fact will be established
by the isotropy of $\bm{C}$. Let $\bm{P}_{\cdot,j}$ denote the $j^{\mathrm{th}}$
column of $\bm{P}$. Since $\bm{P}$ is full rank, $\bm{P}_{\cdot,j}\ne\bm{0}$
for all $j=1,2,\cdots,n$. By assumption, $\bm{C}$ is isotropically
distributed. This implies that $\bm{C}\bm{P}_{\cdot,j}\ne0$ with
probability one \cite{James1954_Normal_Multivariate_Analysis_Orthogonal_Group}.
$\bm{CP}$ is composed of finite columns. It follows that with probability
one, no entry of $\bm{CP}$ is zero.

So far, we have proved that $\bm{O}_{t}$ has full row rank with probability
one if the Jordan normal form $\bm{J}=\bm{P}^{-1}\bm{A}\bm{P}$ is
diagonal. Note that by assumption, the Jordan normal form is diagonal
with probability one. We have $\mathrm{rank}\left(\bm{O}_{\bm{t}}\right)=m\ge O\left(K\log\frac{n}{K}\right)$
with probability one. This proves this corollary.

\section{\label{sec:Conclusion}Concluding Remarks}

In this paper we obtained sufficiency conditions for the observability
of a linear system where the number of non-zeros in the initial states
is known to be less than the dimensionality of the system. The discussion
also applies to the case if certain elements have known values and
we wish to reconstruct the unknown values.

Two models were included; one is for a deterministic model and the
other for a stochastic model. We observed that a much lower number
of observations (even when the observations are randomly picked) can
be used to recover the initial condition. Furthermore, this can be
done by a linear or quadratic program.

An interesting extension of this problem is for the case when there
are some non-zero terms but terms which are known to have small magnitude,
that is a robust formulation of initial condition recovery when the
disturbance is an $l_{2}$ ball of small radius.

Compressive sensing offers new directions for design of information
structures in networked control systems. Recent work \cite{BasarACC2011}
lays out designs based on compressive sensing principles for such
systems. We believe there will be further results specific to control
systems, in particular on the inherent interaction between estimation
and control in decentralized control systems.  \bibliographystyle{ieeetr}
\bibliography{main2}

\end{document}